\documentclass[conference]{IEEEtran}
\usepackage[letterpaper, left=0.55in, right=0.55in, bottom=1in, top=0.75in]{geometry}

\usepackage{ragged2e}
\usepackage{floatrow}
\usepackage{amsthm}
\usepackage[utf8]{inputenc}
\usepackage[english]{babel}
\theoremstyle{remark}
\newtheorem{theorem}{Theorem}

\newtheorem{Proposition}{Proposition}

\usepackage{multicol}
\usepackage{epsfig}
\usepackage{epstopdf}
\usepackage{float}
\usepackage{perpage}
\MakeSorted{figure}
\MakeSorted{table}
\usepackage{array}
\usepackage{graphicx}
\usepackage{amsfonts}
\usepackage{amssymb}
\let\labelindent\relax
\usepackage{enumitem}
\usepackage[english]{babel}
\usepackage[utf8]{inputenc}
\usepackage[linesnumbered,ruled,vlined]{algorithm2e}

\usepackage{cite}
\DeclareMathAlphabet\mathbfcal{OMS}{cmsy}{b}{n}
\ifCLASSINFOpdf
\else
\fi
\usepackage[cmex10]{amsmath}
\usepackage{array}
\newcounter{mytempeqncnt}
\usepackage[linesnumbered,ruled,vlined]{algorithm2e}
\usepackage[font=footnotesize,labelfont=bf, figurename=Fig.]{caption} 
\usepackage{subcaption}
\usepackage{fix-cm}
\usepackage{etoolbox} 
\newtoggle{SINGLE_COL}
\toggletrue{SINGLE_COL}     
\togglefalse{SINGLE_COL}   

\newtoggle{BIG_EQUATION}
\toggletrue{BIG_EQUATION}  

\usepackage[pscoord]{eso-pic}

\begin{document}

\title{Multi-Pair Two Way AF Full-Duplex Massive MIMO Relaying with ZFR/ZFT Processing\\[-1ex]}
\author{
Ekant Sharma, Rohit Budhiraja and K Vasudevan\\
Department of Electrical Engineering, IIT Kanpur, India\\
email: \{ekant, rohitbr, vasu\}@iitk.ac.in
       \\[-4ex]}
\maketitle
\begin{abstract}
We consider two-way amplify and forward  relaying, where multiple full-duplex user pairs  exchange information via a shared full-duplex massive multiple-input multiple-output (MIMO) relay. We derive closed-form lower bound for the spectral efficiency with zero-forcing processing at the relay, by using minimum mean squared error channel estimation. The zero-forcing lower bound for the system model considered herein, which is valid for arbitrary number of antennas, is not yet derived in the massive MIMO relaying literature. We numerically demonstrate the accuracy of the derived lower bound and the performance improvement achieved using zero-forcing processing. We also numerically demonstrate the spectral gains achieved by a full-duplex system over a half-duplex one for various antenna regimes.
\end{abstract}
\iftoggle{SINGLE_COL}{\vspace*{-0.25in}}{}
\begin{IEEEkeywords}\iftoggle{SINGLE_COL}{\vspace*{-0.15in}}{}
Full-duplex, relay, spectral efficiency.
\end{IEEEkeywords}
\IEEEpeerreviewmaketitle
\iftoggle{SINGLE_COL}{}{\vspace*{-0.1in}}
\section{Introduction}\iftoggle{SINGLE_COL}{}{\vspace*{-0.05in}}
Relay based communication is being extensively investigated to expand the coverage, improve the diversity, increase the data rate and reduce the power consumption of wireless communication systems \cite{DBLP:journals/corr/abs-1303-2817,DBLP:conf/iswcs/ZhangTH13a}. The current generation relays are mostly half-duplex due to their implementation simplicity. 
Full-duplex technology is becoming  popular after recent studies, e.g., \cite{DBLP:journals/corr/NadhSSAG16, fd_tut_ref_ashu}, demonstrated a significant reduction in the loop interference, caused due to transmission and reception on the same channel. A full-duplex relay~\cite{fd_relay_si_can_young}, commonly known as full-duplex one-way relay, transmits and receives on the same channel, and can theoretically double the spectral efficiency, when compared with a half-duplex one-way relay \cite{DBLP:journals/corr/abs-1303-2817,DBLP:conf/iswcs/ZhangTH13a}. 

Full-duplex two-way relaying \cite{DBLP:journals/twc/ZhangMDXK16}, wherein two users exchange two data units in one channel use via a relay, further improves the spectral efficiency. Two-way full-duplex relaying is recently extended to multi-pair two -way full-duplex relaying \cite{DBLP:journals/jsac/NgoSML14,DBLP:journals/jsac/ZhangCSX16,Zhang2016Chen} wherein multiple user pairs exchange data via a shared relay in a single channel use. A multi-pair two-way full-duplex relay system has following interference sources: i) co-channel (inter-pair) interference due to multiple users simultaneously accessing the channel; ii) loop interference at the relay and at the users; and iii) inter-user interference caused due to simultaneous transmission and reception by full-duplex~nodes. 

Massive multiple-input multiple-output (MIMO) systems have become popular as they cancel co-channel interference by using simple linear transmit processing schemes e.g., zero-forcing transmission (ZFT) and  maximal-ratio transmission (MRT)  \cite{DBLP:journals/twc/Marzetta10,DBLP:journals/tcom/NgoLM13}, and  significantly improve the spectral efficiency.
Massive MIMO technology is also being incorporated in multi-pair full-duplex relays to cancel the loop interference at the relay, and inter-pair co-channel interference \cite{DBLP:journals/jsac/NgoSML14,DBLP:journals/jsac/ZhangCSX16,Zhang2016Chen,DBLP:journals/twc/DaiD16,mm_relay_hong}. Reference~\cite{DBLP:journals/jsac/NgoSML14} derived the achievable rate and a power allocation scheme to maximize the ergodic sum-rate for one-way decode and forward full-duplex massive MIMO relaying.  Zhang~\textit{et~al.} in \cite{DBLP:journals/jsac/ZhangCSX16} proposed four power scaling schemes for two-way full-duplex massive MIMO relaying to improve its spectral and energy efficiency. Reference~\cite{Zhang2016Chen} developed a power allocation scheme to maximize the sum-rate for multi-pair two-way full-duplex massive MIMO amplify-and-forward (AF)  relaying by using maximal-ratio combining (MRC)/MRT processing at the relay, and by using least squares (LS) channel estimation.
Dai~\textit{et~al.} in \cite{DBLP:journals/twc/DaiD16} considered a half-duplex multi-pair two-way massive MIMO  AF relay and derived closed-form achievable rate expressions and a power allocation scheme to maximize the sum-rate with imperfect channel state information (CSI). The authors in \cite{mm_relay_hong} developed power scaling schemes for half-duplex massive MIMO one-way relay systems.

The authors in \cite{Zhang2016Chen} have derived the spectral-efficiency lower bound for MRC/MRT relay processing with LS channel estimation. We extend the work done in \cite{Zhang2016Chen}, and next list the main contributions of this paper.
\begin{itemize}
\item We derive closed-form lower bound for the spectral efficiency of the multi-pair two-way AF full-duplex massive MIMO relay for arbitrary number of relay antennas. We consider zero-forcing reception (ZFR)/zero-forcing transmission (ZFT) processing at the relay and minimum-mean-square-error (MMSE) relay channel estimation. We note that the bound obtained for MRC/MRT processing based on LS channel estimation in \cite{Zhang2016Chen} cannot be trivially extended to the ZFR/ZFT processing with MMSE channel estimation, considered herein. This closed-form spectral-efficiency lower bound, with arbitrary number of relay antennas, to the best of our knowledge, have not yet been derived in the massive MIMO relaying literature.
\item We also numerically demonstrate the considerable spectral efficiency gains achieved due to MMSE channel estimation and ZFR/ZFT processing. 
\end{itemize}


\iftoggle{SINGLE_COL}{}{\vspace*{-0.05in}}
\section{System Model}
\label{sys_model}\iftoggle{SINGLE_COL}{}{\vspace*{-0.06in}}
We consider multi-pair two-way AF full-duplex relaying as shown in Fig.~\ref{system_model}, where $K$ full-duplex user pairs communicate via a single full-duplex relay on the same time-frequency resource. We assume that the user $S_{2m-1}$ for $m=1$ to $K$ on one side of the relay, wants to send as well as receive from the user $S_{2m}$ that is on the other side of the relay. We also assume that there is no direct link between the user-pairs $(S_{2m-1},S_{2m})$ on the either side of the relay due to large path loss and heavy shadowing. Also the relay has $N$ transmit and $N$ receive antennas, while each user has one transmit and one receive antenna. The users on either side of the relay, due to full-duplex architecture, interfere with each other; the interference caused is termed as inter-user interference.

\begin{figure}[!htb]
	\centering
	\includegraphics[scale=\iftoggle{SINGLE_COL}{0.35}{0.35}]{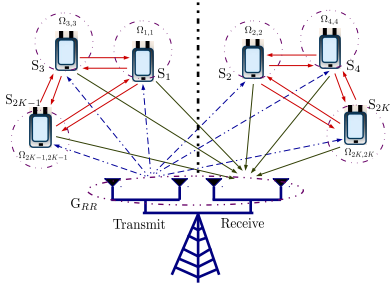}\iftoggle{SINGLE_COL}{}{\vspace*{-8pt}}
	\caption{Multi-pair two-way full-duplex AF massive MIMO relay system.
}
	\label{system_model}	
\end{figure}\vspace*{-4pt}

At time instant $n$, each user $S_{k}$, $k =1$ to $2K$, transmits the signal $\sqrt{p_{k}}x_{k}(n)$ to the relay, and simultaneously the relay broadcasts $\mathbf{x}_{R}(n)\in\mathbb{C}^{N\times 1}$ to all users. Here $p_k$ denotes the transmit power of the $k$th user. The received signal at the relay and the user $S_{k}$ are given as\iftoggle{SINGLE_COL}{}{\vspace*{-0.02in}}
\begin{align}
\mathbf{y}_{R}(n) &= \sum_{k=1}^{2K}\sqrt{p_{k}} \mathbf{g}_{k}x_{k}(n) +\mathbf{G_{RR}}{\mathbf{x}}_{R}(n)+\mathbf{z}_{R}(n)\label{yrn}\iftoggle{SINGLE_COL}{}{\nonumber\\
&}= \tilde{\mathbf{G}}\mathbf{x}(n) +\mathbf{G_{RR}}{\mathbf{x}}_{R}(n)+\mathbf{z}_{R}(n),\\
y_{k}(n)&= \mathbf{f}_{k}^{T}\mathbf{x}_{R}(n)+\sum_{i,k\in U_{k}}\Omega_{k,i}\sqrt{p_{i}}x_{i}(n)+z_{k}(n)\label{ykn}.
\end{align}
We denote the matrix $\mathbf{G} = \left[\mathbf{g}_{1},\,\mathbf{g}_{2},\,\mathbf{g}_{3},\,\cdots,\, \mathbf{g}_{2K}\right]\in\mathbb{C}^{N\times 2K}$ and the matrix (to be used later in the sequel) $\mathbf{F}=\left[\mathbf{f}_{1},\,\mathbf{f}_{2},\,\mathbf{f}_{3},\,\cdots,\, \mathbf{f}_{2K}\right]\in\mathbb{C}^{N\times 2K}$, where $\mathbf{g}_k$ and $\mathbf{f}_k$ denote the channels from the transmit antenna of the $k$th user to the relay receive antenna array, and from the relay's  transmit antenna array to the receive antenna of the $k$th user, respectively. Further, $\tilde{\mathbf{G}} = \mathbf{G}\mathbf{P}$ where $\mathbf{P}=\mbox{diag}\left\{\sqrt{p_{1}},\,\sqrt{p_{2}},\cdots\,\sqrt{p_{2K}}\right\}$ with $0\leq p_{i}\leq P_{0}$ and $\sum_{i=1}^{2K}p_{i} = P$. The receive signal at the relay and the users are interfered by their own transmit signal, which is called as the self-loop interference. Here $\mathbf{G}_{RR}$ and ${\Omega}_{k,k}$ denote the self-loop interference at the relay and the user $S_{k}$. The entries of the matrix  $\mathbf{G}_{RR}$ and the scalar ${\Omega}_{k,k}$ are independent and identically distributed (i.i.d.) with distribution $\mathcal{C}\mathcal{N}(0,\sigma_{LIR}^{2})$ and $\mathcal{C}\mathcal{N}(0,\sigma_{k,k}^2)$, respectively. The term $\Omega_{k,i}\, (k,i\in U_{k},i\neq k)$ denote the inter-user interference channel, which is modeled as i.i.d. $\mathcal{C}\mathcal{N}(0,\sigma_{k,i}^{2})$, where the set $U_{k}=\left[1,3,5,\cdots,2K-1\right]$ for odd $k$ and $U_{k}=\left[2,4,6,\cdots,2K\right]$ for even $k$. The vector $\mathbf{x}(n)= \left[x_{1}(n),\,x_{2}(n),\,x_{3}(n),\,\cdots,\,x_{2K}(n)\right]^{T}\in\mathbb{C}^{N\times 1}$ with $\mathbb{E}\left[\mathbf{x}(n)\mathbf{x}^{H}(n)\right]= \mathbf{I}_{2K}$. The vector $\mathbf{z}_{R}(n)\in\mathbb{C}^{N\times 1}$ and the scalar ${z}_{k}(n)$ are additive white Gaussian noise (AWGN) at the relay and the user $S_{k}$. The elements of $\mathbf{z}_{R}$ and the scalar ${z}_{k}(n)$ are modeled  as i.i.d. $\mathcal{C}\mathcal{N}(0,\sigma_{nr}^2)$ and $\mathcal{C}\mathcal{N}(0,\sigma_{n}^2)$, respectively.

The channel matrices account for both small-scale and large-scale fading; we therefore express $\mathbf{G} = \mathbf{H}_{u}\mathbf{D}_{u}^{1/2}$ and $\mathbf{F} = \mathbf{H}_{d}\mathbf{D}_{d}^{1/2}$. Here the small-scale fading matrices $\mathbf{H}_{u}$ and $\mathbf{H}_{d}$ have i.i.d. $\mathcal{C}\mathcal{N}(0,1)$ elements, while the $k$th element of large-scale diagonal fading matrices $\mathbf{D}_{u}$ and $\mathbf{D}_{d}$ are denoted as $\sigma_{g,k}^{2}$ and $\sigma_{f,k}^{2}$, respectively.
The inter-element distance is assumed to be smaller than the distance between the transmit array and the receive array which leads to the channel between the transmit and receive antennas to be independent.

In the first time slot $(n=1)$, the relay only receives the signal and does not transmit. The signals received at the relay and the user $S_{k}$ are given respectively as 
\begin{eqnarray}
\mathbf{y}_{R}(1) &=& \tilde{\mathbf{G}}\mathbf{x}(1) +\mathbf{z}_{R}(1)\iftoggle{SINGLE_COL}{}{\nonumber\\}\iftoggle{BIG_EQUATION}{}{,\quad}
y_{k}(1)\iftoggle{SINGLE_COL}{}{&=&} \sum_{i,k\in U_{k}}\Omega_{k,i}\sqrt{p_{i}}x_{i}(1)+z_{k}(1).
\end{eqnarray}
At the $n$th time slot, the relay linearly precodes its received signal $\mathbf{y}_{R}(n-1)$ using a matrix $\mathbf{W}$ such that
\iftoggle{SINGLE_COL}{\vspace*{-0.15in}}{}\begin{align}\label{xrn}
\mathbf{x}_{R}(n) &= \alpha \mathbf{W}\mathbf{y}_{R}(n-1),
\end{align}
where $\alpha$ is the scaling factor chosen to satisfy the relay power constraint. The relay transmit signal $\mathbf{x}_{R}$, similar to \cite{DBLP:journals/jsac/ZhangCSX16}, can be re-expressed using (\ref{yrn}) as 
\begin{align}
\hspace{-0.1in}\mathbf{x}_{R}(n) &= s\left(\mathbf{x}(n-\nu)+\mathbf{x}(n-2\nu)+\cdots\iftoggle{SINGLE_COL}{}{\right.\nonumber\\&\left.\,\,\,}
+\,\mathbf{z}_{R}(n-\nu)+\mathbf{z}_{R}(n-2\nu)+\cdots \right).
\end{align}
Here $s(\cdot)$ is a function involving vector and matrix operation, and $\nu$ is the relay processing  delay ($\nu = 1$ in this paper).
There are various solution proposed in the relaying literature, e.g., \cite{riihonen2009spatial}, that significantly suppress the self-loop interference caused due to $\mathbf{x}_{R}$  such that the residual self-loop interference can be replaced with $\tilde{\mathbf{x}}_{R}(n)$, an additional Gaussian noise source with $\left\{\mathbb{E}\left[\tilde{\mathbf{x}}_{R}(n)\tilde{\mathbf{x}}_{R}^{H}(n)\right]\right\} =\frac{P_{R}}{N}\mathbf{I}_N$ \cite{DBLP:journals/jsac/ZhangCSX16}. Therefore, the relay receive signal in \eqref{yrn} can be re-expressed~as\iftoggle{SINGLE_COL}{\vspace*{-0.15in}}{}
\begin{align}\label{modyr}
\tilde{\mathbf{y}}_{R}(n) &= \tilde{\mathbf{G}}\mathbf{x}(n) +\mathbf{G_{RR}}\tilde{\mathbf{x}}_{R}(n)+\mathbf{z}_{R}(n).
\end{align}
We re-write the relay transmit signal in (\ref{xrn}), using (\ref{modyr}),  as
\begin{align}\label{modixrn2}
\mathbf{x}_{R} &= \alpha \mathbf{W}\tilde{\mathbf{y}}_{R}(n-1).
\end{align}
For the sake of brevity, we will drop the time labels. Using (\ref{modyr}), we re-write (\ref{xrn}) as
\begin{align}\label{modixrn2}
\mathbf{x}_{R} &= \alpha \mathbf{W}\tilde{\mathbf{y}}_{R}= \alpha \mathbf{W}\tilde{\mathbf{G}}\mathbf{x} + \alpha \mathbf{W} \mathbf{G_{RR}}\tilde{\mathbf{x}}_{R} + \alpha \mathbf{W}\mathbf{z}_{R}.
\end{align}
The relay transmit signal should satisfy its transmit power constraint such that
\begin{align}\label{exppr}
\hspace{-0.1in}\mathbf{P}_{R} &= \mbox{Tr}\left\{\mathbb{E}\left[\mathbf{x}_{R}\mathbf{x}^{H}_{R}\right]\right\}\iftoggle{SINGLE_COL}{}{\nonumber\\
&}\hspace{-0.2in}= \hspace{-0.03in}\mathbb{E}\hspace{-0.03in}\left[\hspace{-0.03in}\|\alpha \mathbf{W}\tilde{\mathbf{G}}\mathbf{x}\|^{2}\hspace{-0.02in}\right]\hspace{-0.05in}+\hspace{-0.03in}\mathbb{E}\left[\|\alpha \mathbf{W} \mathbf{G_{RR}}\tilde{\mathbf{x}}_{R}\|^{2}\right]\hspace{-0.03in}+\hspace{-0.03in}\mathbb{E}\left[\|\alpha \mathbf{W}\mathbf{z}_{R}\|^{2}\right]\hspace{-0.03in},
\end{align}
which leads to the following value of the scaling factor $\alpha$
\begin{align}\label{alpha}
\hspace{-0.1in}\alpha \hspace{-0.03in}=\hspace{-0.05in} \sqrt{\frac{P_{R}}{\mathbb{E}\left[\|\mathbf{W}\tilde{\mathbf{G}}\mathbf{x}\|^{2}\right]\hspace{-0.05in}+\hspace{-0.05in}\mathbb{E}\left[\| \mathbf{W} \mathbf{G_{RR}}\tilde{\mathbf{x}}_{R}\|^{2}\right]\hspace{-0.05in}+\hspace{-0.05in}\mathbb{E}\left[\| \mathbf{W}\mathbf{z}_{R}\|^{2}\right]}}.
\end{align}
We next re-express the received signal after self-interference cancellation (SIC) at the user $S_{k}$ given in (\ref{ykn}), using (\ref{modixrn2}), as
\begin{eqnarray}\label{yktilde}
\tilde{y}_{k} &=& \underbrace{\alpha \mathbf{f}_{k}^{T}\mathbf{W}\sqrt{p_{k^{'}}}{\mathbf{g}_{k^{'}}}x_{k^{'}}}_{\text{desired signal}} +\underbrace{\alpha\sqrt{p_{k}}\lambda_{k}x_{k}}_{\text{residual interference}}\iftoggle{SINGLE_COL}{}{\nonumber\\[-2pt]
&{+}&} \underbrace{\alpha \mathbf{f}_{k}^{T}\mathbf{W}\sum\limits_{i\neq k,k^{'}}^{2K}\sqrt{p_{i}}\mathbf{g}_{i}x_{i}}_{\text{inter-pair interference}}+\underbrace{\alpha \mathbf{f}_{k}^{T}\mathbf{W} \mathbf{G_{RR}}\tilde{\mathbf{x}}_{R}}_{\text{amplified loop interference}}\nonumber\\
&{+}&\hspace{-0.2in} \underbrace{\alpha \mathbf{f}_{k}^{T}\mathbf{W}\mathbf{z}_{R}}_{\text{amplified noise from relay}}\hspace{-0.1in}+\underbrace{\sum\limits_{i,k\in U{k}}\Omega_{k,i}\sqrt{p_{i}}x_{i}}_{\substack{\text{self loop interference}\\ \text{and inter-user interference}}} + \hspace{-0.1in}\underbrace{z_{k}}_{\text{AWGN at $S_{k}$}}\hspace{-0.2in}.
\end{eqnarray}
Here $(k,k^{'})=(2m-1,2m)$ or $(2m,2m-1)$, for $m=1,2,3,\cdots,K$ denotes the user pair which exchange information with one another. The scalar $\lambda_{k} = {\mathbf{f}}_{k}^{T}\mathbf{W}{{\mathbf{g}_{k}}}-\hat{\mathbf{f}}_{k}^{T}\mathbf{W}{\hat{\mathbf{g}_{k}}}$ is the residual self-interference.
In this work we assume that the relay estimates channels $\mathbf{G}$ and $\mathbf{F}$ and uses them to design the precoder $\mathbf{W}$. The relay then transmits the SIC coefficient $\hat{\mathbf{f}}_{k}^T\mathbf{W}\hat{\mathbf{g}}_{k}$ for each user, where $\hat{\mathbf{f}}_{k}$ and $\hat{\mathbf{g}}_{k}$ are the estimated channel coefficients.
Before designing the relay precoder $\mathbf{W}$, we briefly digress to discuss the MMSE channel estimation~process.
\iftoggle{SINGLE_COL}{}{\vspace*{-0.1in}}
\section{Channel Estimation}
\label{ch_est_ref}\iftoggle{SINGLE_COL}{}{\vspace*{-0.1in}}
Assuming the coherence interval for transmission to be $T$ symbols, all users simultaneously transmit pilot sequence of length $\tau\leq T$ symbols to the relay. During pilot transmission phase, the relay will receive the following signal at the receive and transmit antenna array
\begin{align}
\hspace{-0.11in}\mathbf{Y}_{R,R}\hspace{-0.03in} &=\hspace{-0.03in} \sqrt{\tau P_{\rho}}\mathbf{G}\mathbf{\varphi}\hspace{-0.03in} +\hspace{-0.03in} \mathbf{N}_{R,R},\hspace{-0.03in}\mbox{ and }\hspace{-0.03in} \mathbf{Y}_{R,T} \iftoggle{BIG_EQUATION}{}{&}\hspace{-0.03in}=\hspace{-0.03in} \sqrt{\tau P_{\rho}}\mathbf{{F}}{\mathbf{\varphi}} \hspace{-0.03in}+\hspace{-0.03in} \mathbf{N}_{R,T}
\end{align}
where $\sqrt{\tau P_{\rho}}\mathbf{\varphi}\in\mathbb{C}^{2K\times \tau}$ denotes the pilot symbols transmitted from all users with $P_{\rho}$ being the transmit power of each pilot symbol. The matrices $\mathbf{N}_{R,R}$ and $\mathbf{N}_{R,T}$ denote the AWGN noise whose elements are i.i.d. and distributed as $\mathcal{C}\mathcal{N}(0,1)$. The pilots are assumed to be orthogonal such that $\mathbf{\varphi}\mathbf{\varphi}^{H}= \mathbf{I}_{2K}$ for $\tau \geq 2K$\cite{Biguesh}. The MMSE channel estimate of $\mathbf{G}$ and $\mathbf{F}$ are given by~\cite{kay1993fundamentals}\iftoggle{SINGLE_COL}{\vspace*{-0.2in}}{}
\begin{align}
\hat{\mathbf{G}}&= \frac{1}{\sqrt{\tau P_{\rho}}}\mathbf{Y}_{R,R}\mathbf{\varphi}^{H}\bar{\mathbf{D}}_{u}=\mathbf{G}\bar{\mathbf{D}}_{u}+\frac{\mathbf{N}_{R,R}\mathbf{\varphi}^{H}}{\sqrt{\tau P_{\rho}}}\bar{\mathbf{D}}_{u}\nonumber\\
\hat{\mathbf{F}}&= \frac{1}{\sqrt{\tau P_{\rho}}}\mathbf{Y}_{R,T}{\mathbf{\varphi}}^{H}\bar{\mathbf{D}}_{d}=\mathbf{F}\bar{\mathbf{D}}_{d}+\frac{\mathbf{N}_{R,T}\mathbf{\varphi}^{H}}{\sqrt{\tau P_{\rho}}}\bar{\mathbf{D}}_{d}.\end{align}
The matrices $\hat{\mathbf{G}}=\left[\hat{\mathbf{g}}_{1},\,\hat{\mathbf{g}}_{2},\,\cdots,\,\hat{\mathbf{g}}_{2K}\right]\in\mathbb{C}^{N\times 2K}$ and $\hat{\mathbf{F}}=\left[\hat{\mathbf{f}}_{1},\,\hat{\mathbf{f}}_{2},\,\cdots,\,\hat{\mathbf{f}}_{2K}\right]\in\mathbb{C}^{N\times 2K}$, $\bar{\mathbf{D}}_{u}=\left(\frac{\mathbf{D}_{u}^{-1}}{{\tau P_{\rho}}}+\mathbf{I}_{2K}\right)^{-1}$ and $\bar{\mathbf{D}}_{d}=\left(\frac{\mathbf{D}_{d}^{-1}}{{\tau P_{\rho}}}+\mathbf{I}_{2K}\right)^{-1}$. 

We note that the elements of $\mathbf{N}^{H}_{R,R}\mathbf{\varphi}$ and $\tilde{\mathbf{N}}^{H}_{R,T}\mathbf{\varphi}$ are  distributed as $\mathcal{C}\mathcal{N}(0,1)$. We therefore have
$\hat{\mathbf{G}} = \mathbf{G}-\mathbf{E}_{g} \text{ and }\hat{\mathbf{F}} = \mathbf{F}-\mathbf{E}_{f}$,
where $\mathbf{E}_{g}$ and $\mathbf{E}_{f}$ are estimation error matrices. The channel matrices $\hat{\mathbf{G}}$ and $\hat{\mathbf{F}}$ are independent of the error matrices ${\mathbf{E}_{g}}$ and ${\mathbf{E}_{f}}$, respectively\cite{kay1993fundamentals}. The matrices $\hat{\mathbf{G}}$ and $\hat{\mathbf{F}}$ are distributed as $\mathcal{C}\mathcal{N}(0,\hat{\mathbf{D}}_{u})$ and $\mathcal{C}\mathcal{N}(0,\hat{\mathbf{D}}_{d})$, respectively; the matrices  $\hat{\mathbf{D}}_{u}=\mbox{diag}\left\{\hat\sigma_{g,1}^{2},\,\hat\sigma_{g,2}^{2},\,\cdots,\,\hat\sigma_{g,2K}^{2}\right\}$ and $\hat{\mathbf{D}}_{d}=\mbox{diag}\left\{\hat\sigma_{f,1}^{2},\,\hat\sigma_{f,2}^{2},\,\cdots,\,\hat\sigma_{f,2K}^{2}\right\}$, with $\hat\sigma_{g,k}^{2}=\frac{\tau P_{\rho}\sigma_{g,k}^{4}}{\tau P_{\rho}\sigma_{g,k}^{2}+1}$ and $\hat\sigma_{f,k}^{2}=\frac{\tau P_{\rho}\sigma_{f,k}^{4}}{\tau P_{\rho}\sigma_{f,k}^{2}+1}$. Hence, ${\mathbf{E_{g}}}\sim \mathcal{C}\mathcal{N}(0,\mathbf{D}_{u}-\hat{\mathbf{D}}_{u})$ and ${\mathbf{E_{f}}}\sim \mathcal{C}\mathcal{N}(0,\mathbf{D}_{d}-\hat{\mathbf{D}}_{d})$, with $\mathbf{D}_{u}-\hat{\mathbf{D}}_{u} = \mbox{diag}\left\{\sigma_{\xi,g,1}^{2},\,\sigma_{\xi,g,2}^{2},\,\sigma_{\xi,g,3}^{2},\,\cdots,\,\sigma_{\xi,g,2K}^{2}\right\}$ and $\mathbf{D}_{d}-\hat{\mathbf{D}}_{d} = \mbox{diag}\left\{\sigma_{\xi,f,1}^{2},\,\sigma_{\xi,f,2}^{2},\,\sigma_{\xi,f,3}^{2},\,\cdots,\,\sigma_{\xi,f,2K}^{2}\right\}$ with  $\sigma_{\xi,g,k}^{2}=\frac{\sigma_{g,k}^{2}}{\tau P_{\rho}\sigma_{g,k}^{2}+1}$ and $\sigma_{\xi,f,k}^{2}=\frac{\sigma_{f,k}^{2}}{\tau P_{\rho}\sigma_{f,k}^{2}+1}$.
\section{Relay precoder design}
\label{relay_pre_des}\iftoggle{SINGLE_COL}{}{\vspace*{-0.05in}}
We design relay precoder based on ZFR/ZFT processing. \iftoggle{SINGLE_COL}{}{
\begin{figure*}
\normalsize
\setcounter{mytempeqncnt}{\value{equation}}
\setcounter{equation}{19}
\begin{align}
\mbox{SNR}_{k} &= \frac{\alpha^{2}p_{k^{'}}|{\mathbf{f}}_{k}^{T}\mathbf{W}{{\mathbf{g}_{k^{'}}}}|}{\alpha^{2}p_{k}|\lambda_{k}|^{2}+\alpha^{2}\sum\limits_{i\neq k,k^{'}}^{2K}|{\mathbf{f}}_{k}^{T}\mathbf{W}{{\mathbf{g}_{i}}}|^{2}+\alpha^{2}\|{\mathbf{f}}_{k}^{T}\mathbf{W}\mathbf{G}_{RR}\|^{2}\frac{P_{R}}{N}+\alpha^{2}\|{\mathbf{f}}_{k}^{T}\mathbf{W}\|^{2}\sigma_{n_{R}}^{2}+\sum\limits_{i,k\in U{k}}\sigma_{k,i}^{2}{p_{i}} + \sigma_{n}^{2}}\label{snrk}\setcounter{equation}{19}\\
\mbox{SNR}_{k,\mbox{lower}} &= \frac{\alpha^{2}p_{k^{'}}\left|\mathbb{E}\left[{\mathbf{f}}_{k}^{T}\mathbf{W}{{\mathbf{g}_{k^{'}}}}\right]\right|^{2}}{\alpha^{2}p_{k^{'}}\mbox{var}\left[{\mathbf{f}}_{k}^{T}\mathbf{W}{{\mathbf{g}_{k^{'}}}}\right]+\alpha^{2}p_{k}\mbox{SI}_{k}+\alpha^{2}\mbox{IP}_{k}+\alpha^{2}\mbox{NR}_{k}+\alpha^{2}\mbox{LIR}_{k}+\mbox{UI}_{k}+\mbox{NU}_{k}}\setcounter{equation}{24}\label{gammalower}
\end{align}
\setcounter{equation}{38}
\hrule
\end{figure*}
}
The ZFR/ZFT matrix using estimated CSI is given by
\begin{eqnarray}\label{wzf}\setcounter{equation}{14}
\mathbf{W} &=& \hat{\bar{\mathbf{F}}}^{*}\mathbf{T}\hat{\bar{\mathbf{G}}}^{H},\label{wzf}
\end{eqnarray}
where $\hat{\bar{\mathbf{F}}}=\hat{\mathbf{F}}\left(\hat{\mathbf{F}}^{H}\hat{\mathbf{F}}\right)^{-1}$ and $\hat{\bar{\mathbf{G}}}=\hat{\mathbf{G}}\left(\hat{\mathbf{G}}^{H}\hat{\mathbf{G}}\right)^{-1}$.
We next state the following proposition to simplify the relay scaling factor $\alpha$ in \eqref{alpha}.
\iftoggle{SINGLE_COL}{}{\vspace*{-0.05in}}
\begin{Proposition}\label{pre2}
For ZFR/ZFT precoder
\begin{eqnarray}\label{alphazfproof}
\alpha &=& \displaystyle{\sqrt{\frac{P_{R}}{\hat\lambda +\hat{\eta}\left(\sum\limits_{i=1}^{2K} p_{i}\sigma_{\xi,g,i}^{2}+\sigma_{nr}^{2}+P_{R}\sigma_{LIR}^{2}\right)}}},
\end{eqnarray}
where
$\hat\lambda = \sum\limits_{i=1}^{2K}\frac{p_{i^{'}}}{\left(N-2K-1\right)\hat\sigma_{f,i}^{2}},\, \hat{\eta} = \sum_{j=1}^{2K}\frac{1}{\left(N-2K-1\right)^{2}\hat\sigma_{f,j}^{2}\hat\sigma_{g,	j^{'}}^{2}}.$
\end{Proposition}
\begin{proof}
To derive this result, we will first simplify $\mathbb{E}\left[\|\mathbf{W}\tilde{\mathbf{G}}\mathbf{x}\|^{2}\right]$ using (\ref{wzf}).
\begin{align}
&\mathbb{E}\left[\|\mathbf{W}\tilde{\mathbf{G}}\mathbf{x}\|^{2}\right]=
\mathbb{E}\left[\|\hat{\bar{\mathbf{F}}}^{*}\mathbf{T}\hat{\bar{\mathbf{G}}}^{H}\tilde{\mathbf{G}}\mathbf{x}\|^{2}\right]\nonumber\\
&=\mbox{Tr}\left\{\mathbb{E}\left[(\hat{\bar{\mathbf{F}}}^{*}\mathbf{T}\hat{\bar{\mathbf{G}}}^{H}(\hat{\mathbf{G}}+\mathbf{E}_{g})\mathbf{P}\mathbf{P}^{H}(\hat{\mathbf{G}}^{H}+\mathbf{E}_{g}^{H})\hat{\bar{\mathbf{G}}}\mathbf{T}^{H}\hat{\bar{\mathbf{F}}}^{T})\right]\right\}\nonumber\\
&\stackrel{(a)}{=}\mbox{Tr}\left\{\mathbb{E}\left[\hat{\bar{\mathbf{F}}}^{*}\mathbf{T}\mathbf{P}\mathbf{P}^{H}\mathbf{T}\hat{\bar{\mathbf{F}}}^{T}\right]\right\}\iftoggle{SINGLE_COL}{}{\nonumber\\
&} + \sum_{i=1}^{2K}p_{i}\sigma_{\xi,g,i}^{2}\mbox{Tr}\left\{\mathbb{E}\left[\hat{\bar{\mathbf{G}}}\mathbf{T}^{H}\hat{\bar{\mathbf{F}}}^{T}\hat{\bar{\mathbf{F}}}^{*}\mathbf{T}\hat{\bar{\mathbf{G}}}^{H}\right]\right\}\nonumber\\
&\stackrel{(b)}{=} \sum_{i=1}^{2K}p_{i^{'}}\mathbb{E}\left[\hat{\bar{\mathbf{f}}}_{i}^{H}\hat{\bar{\mathbf{f}}}_{i}\right]+ \sum_{i=1}^{2K}p_{i}\sigma_{\xi,g,i}^{2}\mbox{Tr}\left\{\mathbb{E}\left[\hat{\mathbf{\Lambda}}_{F}^{*}\mathbf{T}\hat{\mathbf{\Lambda}}_{G}\mathbf{T}\right]\right\}\iftoggle{SINGLE_COL}{}{\nonumber\\
&}\stackrel{(c)}{=} \hat{\lambda} + \sum_{i=1}^{2K} p_{i}\sigma_{\xi,g,i}^{2}\hat{\eta},\label{firstzfalpha}
\end{align}
The equality  in $(a)$ is obtained by using the fact that $\hat{\bar{\mathbf{G}}}^{H}\hat{\mathbf{G}} = \hat{\mathbf{G}}^{H}\hat{\bar{\mathbf{G}}}=\mathbf{I}_{2K}$, and  $\mathbb{E}\left[{\mathbf{E}_{g}}\mathbf{P}\mathbf{P}^{H}{\mathbf{E}_{g}}^{H}\right]=\sum\limits_{i=1}^{2K}p_{i}\sigma_{\xi,g,i}^{2}\mathbf{I}_{N}$. In equality $(b)$, we define $\hat{\mathbf{\Lambda}}_{F}\triangleq(\hat{\bar{\mathbf{F}}}^{H}\hat{\bar{\mathbf{F}}})=\left({\bar{\mathbf{F}}}^{H}{\bar{\mathbf{F}}}\right)^{-1}$
and $\hat{\mathbf{\Lambda}}_{G}\triangleq\left(\hat{\bar{\mathbf{G}}}^{H}\hat{\bar{\mathbf{G}}}\right)=\left({\bar{\mathbf{G}}}^{H}{\bar{\mathbf{G}}}\right)^{-1}$ and use the fact that $\mathbf{T}\mathbf{P}\mathbf{P}^{H}\mathbf{T}=\mbox{diag}\{p_{2},p_{1},\cdots,p_{2K},p_{2K-1}\}$.
To derive equality in $(c)$, we first note that the random matrices $\hat{\mathbf{\Lambda}}_{F}$ and $\hat{\mathbf{\Lambda}}_{G}$ have inverse Wishart distribution, where $\hat{\mathbf{\Lambda}}_{F}\sim\mathcal{W}^{-1}(\hat{\mathbf{D}}_{d}^{-1},2K)$, $\hat{\mathbf{\Lambda}}_{G}\sim\mathcal{W}^{-1}(\hat{\mathbf{D}}_{u}^{-1},2K)$, with $\hat{w}_{f,i,j} = \left(\hat{\mathbf{\Lambda}}_{F}\right)_{i,j},\,\hat{w}_{g,i,j} = \left(\hat{\mathbf{\Lambda}}_{G}\right)_{i,j},\,\forall i,j = 1,2,3,...,2K$ and $\mathbb{E}\left[\hat{\mathbf{\Lambda}}_{F}\right] = \frac{\hat{\mathbf{D}}_{d}^{-1}}{N-2K-1}$, $\mathbb{E}\left[\hat{\mathbf{\Lambda}}_{G}\right] = \frac{\hat{\mathbf{D}}_{u}^{-1}}{N-2K-1}$ \cite{graczyk2003complex}. We also have $\mathbb{E}\left[\hat{\bar{\mathbf{f}}}_{i}^{H}\hat{\bar{\mathbf{f}}}_{j}\right]=\mathbb{E}\left[\hat{w}_{f,i,j}\right] =\frac{1}{(N-2K-1)\hat\sigma_{f,i}^{2}}, \forall i=j \,\mbox{and}\, 0,\, \mbox{otherwise}$, and similarly, $\mathbb{E}\left[\hat{\bar{\mathbf{g}}}_{i}^{H}\hat{\bar{\mathbf{g}}}_{j}\right]=\mathbb{E}\left[\hat{w}_{g,i,j}\right] =\frac{1}{(N-2K-1)\hat\sigma_{g,i}^{2}},\forall i=j \,\mbox{and}\, 0,\, \mbox{otherwise}$. With the above equalities, we obtain the equality in $(c)$, where $\hat\lambda = \sum_{i=1}^{2K}\frac{p_{i^{'}}}{\left(N-2K-1\right)\hat\sigma_{f,i}^{2}}$. The expression  $\mbox{Tr}\left\{\mathbb{E}\left[\hat{\mathbf{\Lambda}}_{F}^{*}\mathbf{T}\hat{\mathbf{\Lambda}}_{G}\mathbf{T}\right]\right\}$ is simplified~as 
$\mbox{Tr}\left\{\mathbb{E}\left[\hat{\mathbf{\Lambda}}_{F}^{*}\mathbf{T}\hat{\mathbf{\Lambda}}_{G}\mathbf{T}\right]\right\}=\sum_{j=1}^{2K}\frac{1}{\left(N-2K-1\right)^{2}\hat\sigma_{f,j}^{2}\hat\sigma_{g,	j^{'}}^{2}}\triangleq\hat{\eta}.$
On similar lines, we have
\begin{align}\setcounter{equation}{17}\label{Thirdzfalpha}
&\mathbb{E}\left[\|\mathbf{W}\mathbf{G}_{RR}\tilde{\mathbf{x}}_{R}\|^{2}\right]=P_{R}\sigma_{LIR}^{2}\hat{\eta}.
\end{align}
The last term in the denominator of (\ref{alpha}) can be simplified as
\begin{align}\label{seczfalpha}
&\hspace{-0.3in}\mathbb{E}\left[\|\hat{\bar{\mathbf{F}}}^{*}\mathbf{T}\hat{\bar{\mathbf{G}}}^{H}\mathbf{z}_{R}\|^{2}\right]= \sigma_{nr}^{2}\hat{\eta}.
\end{align}
By using (\ref{firstzfalpha}), (\ref{Thirdzfalpha}) and (\ref{seczfalpha}), we get (\ref{alphazfproof}).
\end{proof}

\section{Spectral efficiency of ZFR/ZFT precoder}
\label{rate_ana_ref}
In this section, we calculate lower bound on the instantaneous spectral efficiency for ZFR/ZFT precoder. The instantaneous $\mbox{SNR}_{k}$ at the user $S_{k}$ can be expressed using (\ref{yktilde}) as in \iftoggle{SINGLE_COL}{}{(\ref{snrk}) (shown at the top of this page).}
\iftoggle{BIG_EQUATION}{}{
\begin{align}\label{snrk}
\mbox{SNR}_{k} &= \frac{\alpha^{2}p_{k^{'}}|{\mathbf{f}}_{k}^{T}\mathbf{W}{{\mathbf{g}_{k^{'}}}}|}{\alpha^{2}p_{k}|\lambda_{k}|^{2}+\alpha^{2}\sum\limits_{i\neq k,k^{'}}^{2K}|{\mathbf{f}}_{k}^{T}\mathbf{W}{{\mathbf{g}_{i}}}|^{2}+\alpha^{2}\|{\mathbf{f}}_{k}^{T}\mathbf{W}\mathbf{G}_{RR}\|^{2}\frac{P_{R}}{N}+\alpha^{2}\|{\mathbf{f}}_{k}^{T}\mathbf{W}\|^{2}\sigma_{n_{R}}^{2}+\sum\limits_{i,k\in U{k}}\sigma_{k,i}^{2}{p_{i}} + \sigma_{n}^{2}}.\iftoggle{SINGLE_COL}{}{\nonumber\\}
\end{align}
}
The spectral efficiency of the system which includes the channel estimation overhead~is
\begin{eqnarray}\setcounter{equation}{21}
R = \left(1-\frac{\tau}{T}\right)\mathbb{E}\left\{\sum_{k=1}^{2K}\log_{2}\left(1+\mbox{SNR}_{k}\right)\right\}.\label{sumexact}
\end{eqnarray}
\iftoggle{SINGLE_COL}{}{
\begin{figure*}\setcounter{equation}{26}
\begin{align}
\hspace{-0.4in}\mbox{SNR}_{k}^{\mbox{zf}}
(p_{k},P_{R}) &= \frac{u_{k}p_{k^{'}}}{\displaystyle{\sum_{i=1}^{2K}\hspace{-0.03in}\left(\hspace{-0.03in}d^{(1)}_{k,i}\hspace{-0.01in}+\hspace{-0.01in}d^{(2)}_{k,i}P_{R}^{-1}\hspace{-0.01in}+\hspace{-0.06in}\sum_{i,k\in U_{k}}\hspace{-0.06in}p_{i}P_{R}^{-1}d^{(3)}_{k,i}\hspace{-0.03in}\right)\hspace{-0.03in}p_{i}\hspace{-0.01in}+\hspace{-0.01in}\left(\hspace{-0.01in}v^{(1)}_{k}+v^{(2)}_{k}P_{R}+v^{(3)}_{k}P_{R}^{-1}\right)+\sum_{i,k\in U_{k}}\left(w^{(1)}_{k,i}+P_{R}^{-1}w^{(2)}_{k,i}\right)p_{i}}}\label{gammazft2}
\end{align}
\hrule
\end{figure*}
}

Next we derive a lower bound on the achievable rate using the method in \cite{marzetta2006much},\cite{DBLP:journals/twc/JoseAMV11}. For the $k-k^{'}$ pair, the signal received by the $k$th user can be written as (see \eqref{yktilde}) 
\begin{eqnarray}\setcounter{equation}{22}
\tilde{y}_{k} = {\alpha\sqrt{p_{k^{'}}}\mathbb{E}\left[{\mathbf{f}}_{k}^{T}\mathbf{W}{{\mathbf{g}_{k^{'}}}}\right]x_{k^{'}}} + {\tilde{n}_{k}},
\end{eqnarray}
where \iftoggle{SINGLE_COL}{\vspace*{-0.34in}}{}
\iftoggle{SINGLE_COL}{}{
\begin{eqnarray}
\tilde{n}_{k} &=& \alpha \sqrt{p_{k^{'}}}\left(\mathbf{f}_{k}^{T}\mathbf{W}\mathbf{g}_{k^{'}}-\mathbb{E}\left[{\mathbf{f}}_{k}^{T}\mathbf{W}{{\mathbf{g}_{k^{'}}}}\right]\right)x_{k^{'}} + \alpha\sqrt{p_{k}}\lambda_{k}x_{k}\nonumber\\
&&+\alpha \mathbf{f}_{k}^{T}\mathbf{W}\sum\limits_{i\neq k,k^{'}}^{2K}\sqrt{p_{i}}\mathbf{g}_{i}x_{i} +\alpha \mathbf{f}_{k}^{T}\mathbf{W} \mathbf{G_{RR}}\tilde{\mathbf{x}}_{R}\nonumber\\
&& + \alpha \mathbf{f}_{k}^{T}\mathbf{W}\mathbf{z}_{R}+\sum\limits_{i,k\in U{k}}\Omega_{k,i}\sqrt{p(k)}x_{i} + z_{k}.
\end{eqnarray}
}
\iftoggle{BIG_EQUATION}{}{
\begin{eqnarray}
\tilde{n}_{k} &=& \alpha \sqrt{p_{k^{'}}}\left(\mathbf{f}_{k}^{T}\mathbf{W}\mathbf{g}_{k^{'}}-\mathbb{E}\left[{\mathbf{f}}_{k}^{T}\mathbf{W}{{\mathbf{g}_{k^{'}}}}\right]\right)x_{k^{'}} + \alpha\sqrt{p_{k}}\lambda_{k}x_{k}+\alpha \mathbf{f}_{k}^{T}\mathbf{W}\sum\limits_{i\neq k,k^{'}}^{2K}\sqrt{p_{i}}\mathbf{g}_{i}x_{i} \nonumber\\
&& +\alpha \mathbf{f}_{k}^{T}\mathbf{W} \mathbf{G_{RR}}\tilde{\mathbf{x}}_{R} + \alpha \mathbf{f}_{k}^{T}\mathbf{W}\mathbf{z}_{R}+\sum\limits_{i,k\in U{k}}\Omega_{k,i}\sqrt{p(k)}x_{i} + z_{k}.
\end{eqnarray}
}
The value of $\mathbb{E}\left[{\mathbf{f}}_{k}^{T}\mathbf{W}{{\mathbf{g}_{k^{'}}}}\right]$ can be calculated from the knowledge of channel distribution. We observe that the desired signal and effective noise are uncorrelated. According to \cite{DBLP:journals/tit/HassibiH03,DBLP:journals/tit/Medard00}, we only exploit the knowledge of the $\mathbb{E}\left[{\mathbf{f}}_{k}^{T}\mathbf{W}{{\mathbf{g}_{k^{'}}}}\right]$ in the detection, and treat uncorrelated additive noise $\tilde{n}(k)$ as the worst-case Gaussian noise when computing the spectral efficiency. We, consequently obtain lower bound on the achievable rate~as
\begin{eqnarray}
R_{\mbox{\,lower}} = \left(1-\frac{\tau}{T}\right)\sum\limits_{k=1}^{2K}\log_{2}\left(1+\mbox{SNR}_{k,\mbox{lower}}\right),
\end{eqnarray}
where \iftoggle{SINGLE_COL}{}{$\mbox{SNR}_{k,\mbox{lower}}$ is given by (\ref{gammalower}) as shown at the top of this page.}
\iftoggle{BIG_EQUATION}{}{
\begin{align}\label{gammalower}
\mbox{SNR}_{k,\mbox{lower}} = \frac{\alpha^{2}p_{k^{'}}\left|\mathbb{E}\left[{\mathbf{f}}_{k}^{T}\mathbf{W}{{\mathbf{g}_{k^{'}}}}\right]\right|^{2}}{\alpha^{2}p_{k^{'}}\mbox{var}\left[{\mathbf{f}}_{k}^{T}\mathbf{W}{{\mathbf{g}_{k^{'}}}}\right]+\alpha^{2}p_{k}\mbox{SI}_{k}+\alpha^{2}\mbox{IP}_{k}+\alpha^{2}\mbox{NR}_{k}+\alpha^{2}\mbox{LIR}_{k}+\mbox{UI}_{k}+\mbox{NU}_{k}}.
\end{align}
}
In (\ref{gammalower}), the residual self-interference after SIC $(\mbox{SI})$, the inter-pair interference $(\mbox{IP})$, the amplified noise from the relay $(\mbox{NR})$, amplified loop interference $(\mbox{LIR})$, self-loop interference and inter-user interference $(\mbox{UI})$, and the noise at user $(\mbox{NU})$, are given as following.
\iftoggle{SINGLE_COL}{}{
\setcounter{equation}{25}}
\begin{align}
&\mbox{SI}_{k}= \mathbb{E}\left[|{\mathbf{f}}_{k}^{T}\mathbf{W}{{\mathbf{g}_{k}}}-\hat{\mathbf{f}}_{k}^{T}\mathbf{W}{\hat{\mathbf{g}}_{k}}|^{2}\right],\mbox{IP}_{k}=\hspace{-0.1in}\sum\limits_{i\neq k,k^{'}}^{2K}\hspace{-0.05in}p_{i}\mathbb{E}\left[|{\mathbf{f}}_{k}^{T}\mathbf{W}{{\mathbf{g}_{i}}}|^{2}\right],\nonumber\\
&\mbox{NR}_{k}=\mathbb{E}\left[|{\mathbf{f}}_{k}^{T}\mathbf{W}\mathbf{z}_{R}|^{2}\right],\,\,\mbox{LIR}_{k}= \mathbb{E}\left[|{\mathbf{f}}_{k}^{T}\mathbf{W}\mathbf{G}_{RR}\tilde{\mathbf{x}}|^{2}\right],\nonumber\\
&\mbox{UI}_{k}= \sum_{i,k\in U{k}}p_{i}\mathbb{E}\left[|\Omega_{k,i}{x_{i}}|^{2}\right],\,\,\mbox{NU}_{k}= \mathbb{E}\left[|z_{k}|^{2}\right].
\end{align}

\begin{theorem}\label{theorem2}
The spectral efficiency for a finite number of receive antenna at the relay with imperfect CSI based ZFR/ZFT processing is lower bounded as $\left(1-\frac{\tau}{T}\right) \log_{2}\left\{1+ \mbox{SNR}_{k}^{\mbox{zf}}
(p_{k},P_{R})\right\}$, where \iftoggle{SINGLE_COL}{}{$\mbox{SNR}_{k}^{\mbox{zf}}
(p_{k},P_{R})$ is given by} \iftoggle{SINGLE_COL}{}{(\ref{gammazft2}), shown at the top of this page, with
\begin{align}
d^{(1)}_{k,i} &= \frac{1}{(N-2K-1)}\left(\frac{\sigma_{\xi,f,k}^{2}}{\hat\sigma_{f,i^{'}}^{2}}+ \frac{\sigma_{\xi,g,i}^{2}}{\hat\sigma_{g,k^{'}}^{2}}\right)+ \sigma_{\xi,f,k}^{2}\sigma_{\xi,g,i}^{2}\hat{\eta}\nonumber\\
d^{(2)}_{k,i} &= \sigma_{n}^{2}\left(\frac{1}{\left(N-2K-1\right)\hat\sigma_{f,i^{'}}^{2}} +\hat{\eta}\sigma_{\xi,g,i}^{2}\right)\nonumber\\
d^{(3)}_{k,i} &= \sigma_{k,i}^{2}\left(\frac{1}{\left(N-2K-1\right)\hat\sigma_{f,i^{'}}^{2}} +\hat{\eta}\sigma_{\xi,g,i}^{2}\right)\nonumber\\
v^{(1)}_{k} &= \sigma_{nr}^{2}\left(\frac{1}{(N-2K-1)\hat\sigma_{g,k^{'}}^{2}}+ \sigma_{\xi,f,k}^{2}\hat{\eta}\right)+\hat\eta\sigma_{LIR}^{2}\sigma_{n}^{2}\nonumber\\
v^{(2)}_{k} &= \sigma_{LIR}^{2}\left(\frac{1}{(N-2K-1)\hat\sigma_{g,k^{'}}^{2}}+\sigma_{\xi,f,k}^{2}\hat{\eta}\right)\nonumber\\
v^{(3)}_{k} &= \hat\eta\sigma_{nr}^{2}\sigma_{n}^{2},\,w^{(1)}_{k,i} = \hat\eta\sigma_{k,i}^{2}\sigma_{LIR}^{2},\,w^{(2)}_{k,i} = \hat\eta\sigma_{k,i}^{2}\sigma_{nr}^{2},\,u_{k} = 1\nonumber.
\end{align}
\end{theorem}
\begin{proof}
Refer to Appendix \ref{gamzf}.
\end{proof}
}
\iftoggle{BIG_EQUATION}{}{
\begin{align}\label{gammazft2}
\iftoggle{BIG_EQUATION}{}{&\hspace{-0.3in}}\iftoggle{SINGLE_COL}{}{\hspace{-0.8in}}\mbox{SNR}_{k}^{\mbox{zf}}
(p_{k},P_{R}) \iftoggle{BIG_EQUATION}{}{\nonumber\\
&\hspace{-0.3in}}=  \frac{u_{k}p_{k^{'}}}{\displaystyle{\sum_{i=1}^{2K}\left(d^{(1)}_{k,i}+d^{(2)}_{k,i}P_{R}^{-1}+\sum_{i,k\in U_{k}}p_{i}P_{R}^{-1}d^{(3)}_{k,i}\right)p_{i}+\left(v^{(1)}_{k}+v^{(2)}_{k}P_{R}+v^{(3)}_{k}P_{R}^{-1}\right)+\sum_{i,k\in U_{k}}\left(w^{(1)}_{k,i}+P_{R}^{-1}w^{(2)}_{k,i}\right)p_{i}}},\nonumber\\
\end{align}
where \iftoggle{SINGLE_COL}{
$u_{k}= 1$, $d^{(1)}_{k,i} = \frac{1}{(N-2K-1)}\left(\frac{\sigma_{\xi,f,k}^{2}}{\hat\sigma_{f,i^{'}}^{2}}+ \frac{\sigma_{\xi,g,i}^{2}}{\hat\sigma_{g,k^{'}}^{2}}\right)+ \sigma_{\xi,f,k}^{2}\sigma_{\xi,g,i}^{2}\hat{\eta}$,
$d^{(2)}_{k,i} = \sigma_{n}^{2}\left(\frac{1}{\left(N-2K-1\right)\hat\sigma_{f,i^{'}}^{2}} +\hat{\eta}\sigma_{\xi,g,i}^{2}\right)$, $d^{(3)}_{k,i} = \sigma_{k,i}^{2}\left(\frac{1}{\left(N-2K-1\right)\hat\sigma_{f,i^{'}}^{2}} +\hat{\eta}\sigma_{\xi,g,i}^{2}\right)$, $v^{(1)}_{k} = \sigma_{nr}^{2}\left(\frac{1}{(N-2K-1)\hat\sigma_{g,k^{'}}^{2}}+ \sigma_{\xi,f,k}^{2}\hat{\eta}\right)+\hat\eta\sigma_{LIR}^{2}\sigma_{n}^{2}$, $v^{(2)}_{k} = \sigma_{LIR}^{2}\left(\frac{1}{(N-2K-1)\hat\sigma_{g,k^{'}}^{2}}+\sigma_{\xi,f,k}^{2}\hat{\eta}\right)$, $v^{(3)}_{k} = \hat\eta\sigma_{nr}^{2}\sigma_{n}^{2}$, $w^{(1)}_{k,i} = \hat\eta\sigma_{k,i}^{2}\sigma_{LIR}^{2}$, $w^{(2)}_{k,i} = \hat\eta\sigma_{k,i}^{2}\sigma_{nr}^{2}$}{}
\iftoggle{SINGLE_COL}{}{where 
\begin{eqnarray}
\hspace{0.5in}u_{k} &=& 1,\,\,d^{(1)}_{k,i} = \frac{1}{(N-2K-1)}\left(\frac{\sigma_{\xi,f,k}^{2}}{\hat\sigma_{f,i^{'}}^{2}}+ \frac{\sigma_{\xi,g,i}^{2}}{\hat\sigma_{g,k^{'}}^{2}}\right)+ \sigma_{\xi,f,k}^{2}\sigma_{\xi,g,i}^{2}\hat{\eta},\nonumber\\
d^{(2)}_{k,i} &=& \sigma_{n}^{2}\left(\frac{1}{\left(N-2K-1\right)\hat\sigma_{f,i^{'}}^{2}} +\hat{\eta}\sigma_{\xi,g,i}^{2}\right),\,\,
d^{(3)}_{k,i} = \sigma_{k,i}^{2}\left(\frac{1}{\left(N-2K-1\right)\hat\sigma_{f,i^{'}}^{2}} +\hat{\eta}\sigma_{\xi,g,i}^{2}\right),\nonumber\\
v^{(1)}_{k} &=& \sigma_{nr}^{2}\left(\frac{1}{(N-2K-1)\hat\sigma_{g,k^{'}}^{2}}+ \sigma_{\xi,f,k}^{2}\hat{\eta}\right)+\hat\eta\sigma_{LIR}^{2}\sigma_{n}^{2},\nonumber\\
v^{(2)}_{k} &=& \sigma_{LIR}^{2}\left(\frac{1}{(N-2K-1)\hat\sigma_{g,k^{'}}^{2}}+\sigma_{\xi,f,k}^{2}\hat{\eta}\right),\nonumber\\
v^{(3)}_{k} &=& \hat\eta\sigma_{nr}^{2}\sigma_{n}^{2},\,\,\,w^{(1)}_{k,i} = \hat\eta\sigma_{k,i}^{2}\sigma_{LIR}^{2},\,\,\,w^{(2)}_{k,i} = \hat\eta\sigma_{k,i}^{2}\sigma_{nr}^{2}\nonumber.
\end{eqnarray}}
\end{theorem}
\begin{proof}
Refer to Appendix \ref{gamzf}.
\end{proof}
}

\section{Simulation Results}
\label{simu_sec_ref}
%
\begin{figure}[!htb]
	\centering
	\includegraphics[scale=0.48]{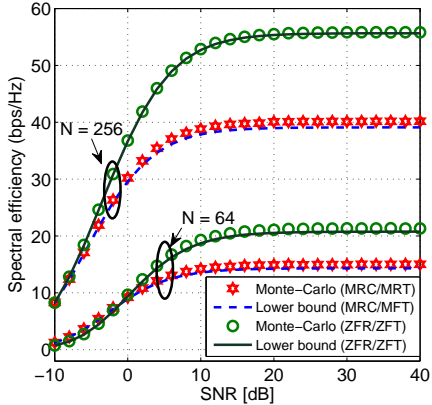}\iftoggle{SINGLE_COL}{}{\vspace*{-8pt}}
	\caption{Spectral efficiency versus $\mbox{SNR}$ for MRC/MRT and ZFR/ZFT, where $\mbox{SNR}_{\rho}=10$~dB.}
	\label{sevssnrmrczf}
\end{figure}

\begin{figure}[!htb]
	\centering
	\includegraphics[scale=0.48]{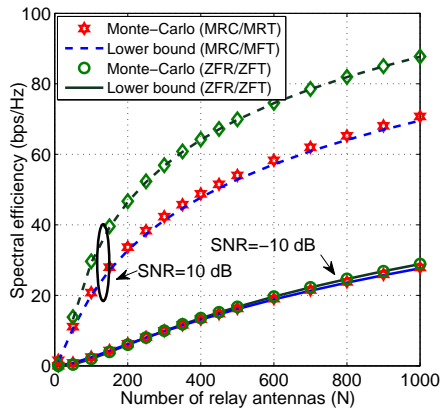}\iftoggle{SINGLE_COL}{}{\vspace*{-8pt}}
	\caption{Spectral efficiency versus the number of relay antennas for MRC/MRT and ZFR/ZFT processing, where $\mbox{SNR}_{\rho}=10$~dB.}
	\label{sevsNmrczf}
\end{figure}

\begin{figure}[!htb]
	\centering
	\includegraphics[scale=0.48]{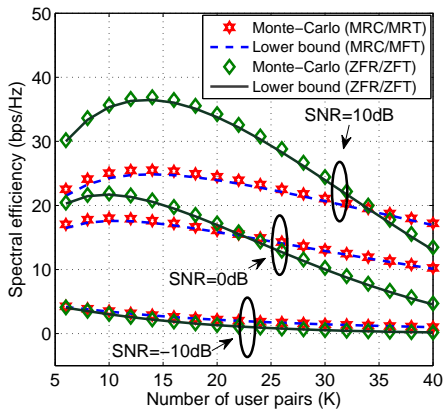}\iftoggle{SINGLE_COL}{}{\vspace*{-8pt}}
	\caption{Spectral efficiency versus number of user pairs for MRC/MRT and ZFR/ZFT processing, where $N = 128,\, \mbox{SNR}_{\rho}=10$~dB.}
	\label{sevsKmrczf}
\end{figure}

\begin{figure}[!htb]
	\centering
	\includegraphics[scale=0.48]{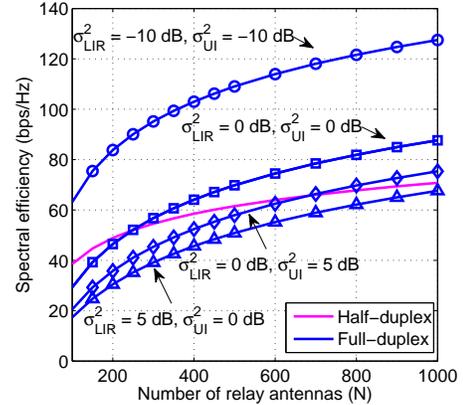}\iftoggle{SINGLE_COL}{}{\vspace*{-8pt}}
	\caption{Spectral efficiency versus number of relay antennas for ZFR/ZFT processing comparing half-duplex and full-duplex systems, where $\mbox{SNR}=10$~dB, $\mbox{SNR}_{\rho}=10$~dB. Here the values of $\sigma_{LIR}^{2}$ and $\sigma_{UI}^{2}$ are with respect to $\sigma^{2}$.}
	\label{sevslirui}
\end{figure}
We investigate the performance of the multi-pair two-way
full duplex AF relay system by using Monte-Carlo simulations. We will validate the lower-bound expression derived for the spectral efficiency in Theorem \ref{theorem2}. In \cite{buddhiraja_ekant_arxiv}, we have derived the lower-bound expression for MRC/MRT processing considering MMSE based channel estimation for the system under consideration. In this paper, we have used the results from \cite{buddhiraja_ekant_arxiv} to compare the performance of ZFR/ZFT with MRC/MRT processing. For this study, we choose, noise variances as $\sigma_{n}^{2}=\sigma_{nr}^{2}=\sigma^{2}$, and the SNR is defined as $\mbox{SNR}=P_{R}/\sigma^{2}$. We define the pilot signal to noise ratio as $\mbox{SNR}_{\rho}=P_{\rho}/\sigma^{2}$, and set the length of the coherence interval $T=200$ symbols, the training length $\tau=2K$. We begin by comparing the analytical lower bound for the spectral efficiency, obtained in Theorem~\ref{theorem2} with their exact expression in \eqref{sumexact} using Monte-Carlo simulations. We compare the bound for $N=64$ and $N=256$ relay antennas and set $K=10$ user pairs, $\sigma^{2}_{g,i}=\sigma^{2}_{f,i}=\sigma^{2}$, for $i=1,2,3,\cdots,2K$, $\sigma_{LIR}^{2}=\sigma^{2}$, $\sigma_{UI}^{2}\triangleq\sigma_{k,j}^{2}=\sigma^{2}$ for $k,j=1,2,3,\cdots,2K$, $\mbox{SNR}_{\rho}=10$~dB and all users are allocated equal power i.e., $p_{i}=P_{R}/2K,\forall i=1,2,3,\cdots,2K$. We see from Fig.~\ref{sevssnrmrczf} that the derived lower bound and exact expression overlap for ZFR/ZFT processing for $N=256$ relay antennas. For MRC/MRT, the lower bound marginally differs from the exact expression. We also observe that the spectral efficiency, for high SNR values, saturates for both MRC/MRT and ZFR/ZFT. This is because the loop interference also increases proportionally with increase in SNR.

Fig.~\ref{sevsNmrczf} compares the spectral efficiency versus $N$ for MRC/MRT and ZFR/ZFT processing with $\mbox{SNR}=10$~dB and $\mbox{SNR}=-10$~dB. The performance of MRC/MRT and ZFR/ZFT processing is almost same for $\mbox{SNR}=-10$~dB. The spectral-efficiency versus $K$ for different value of $\mbox{SNR}$ is shown in Fig.~\ref{sevsKmrczf}. As the number of multi-pairs increases the SNR of each user decreases and hence noise dominates. The ZFR/ZFT neglects the effect of noise which degrades the spectral-efficiency as $K$ increases. In contrast MRC/MRT works well at low SNR as it maximizes the received SNR while neglecting the inter-pair interference.
Fig.~\ref{sevslirui} compares the spectral efficiency versus number of relay antennas for half-duplex and full-duplex system with ZFR/ZFT processing. As we increase the value of  self-loop interference $\sigma_{LIR}^{2}$ and inter-user interference $\sigma_{UI}^{2}$, the spectral efficiency of full-duplex system decreases. For $N<650$, the half-duplex relay with $\sigma_{LIR}^{2}=0$~dB, $\sigma_{UI}^{2}=5$~dB performs better than full-duplex relay. We also observe that with the increase in the the number of relay antennas the rate of increase of spectral efficiency in case of full-duplex relay is higher as compared to half-duplex relay.
\iftoggle{SINGLE_COL}{}{\vspace*{-0.1in}}
\section{Conclusion}
\label{conclude_ref}
We considered a multi-pair AF full-duplex massive MIMO two-way relay with full-duplex users with single transmit and receive antenna. We derived closed-form spectral efficiency expression for ZFR/ZFT relay processing with MMSE channel estimation, and for arbitrary number of relay antennas, which have not yet been derived in the literature. We showed the accuracy of these lower bounds for different number of relay antennas, user pairs and relay transmit power. We also numerically investigated the loop and inter-user interference values for which the full-duplex relay outperforms a half-duplex~relay.
\iftoggle{SINGLE_COL}{}{\vspace*{-0.2in}}
\appendices
\section{}
\label{gamzf}

\iftoggle{SINGLE_COL}{}{
\begin{figure*}
\normalsize
\setcounter{mytempeqncnt}{\value{equation}}
\setcounter{equation}{28}
\begin{align}
&\mbox{var}\left[{\mathbf{f}_{k}}^{T}\mathbf{W}\mathbf{g}_{k^{'}}\right]=\mathbb{E}\left[\left|{\mathbf{f}_{k}}^{T}\mathbf{W}\mathbf{g}_{k^{'}}\right|^{2}\right]-\left|\mathbb{E}\left[{\mathbf{f}_{k}}^{T}\mathbf{W}\mathbf{g}_{k^{'}}\right]\right|^{2}\nonumber\\
&\stackrel{(a)}{=}\sigma_{\xi,f,k}^{2} \mathbb{E}\left[\hat{\mathbf{g}}_{k^{'}}^{H}\mathbf{W}^{H}\mathbf{W}\hat{\mathbf{g}}_{k^{'}}\right]+ \sigma_{\xi,g,k^{'}}^{2}\mathbb{E}\left[\hat{\mathbf{f}}_{k}^{T}\mathbf{W}\mathbf{W}^{H}\hat{\mathbf{f}}_{k}^{*}\right] + \sigma_{\xi,f,k}^{2}\sigma_{\xi,g,k^{'}}^{2}\mbox{Tr}\left\{\mathbb{E}\left[\mathbf{W}\mathbf{W}^{H}\right]\right\}\nonumber\\
&\stackrel{(b)}{=}\sigma_{\xi,f,k}^{2} \mathbb{E}\left[\mathbf{1}_{g,k^{'}}^{T}\mathbf{T}\hat{\mathbf{\Lambda}}_{F}^{*}\mathbf{T}\mathbf{1}_{g,k^{'}}\right]+ \sigma_{\xi,g,k^{'}}^{2} \mathbb{E}\left[\mathbf{1}_{f,k}^{T}\mathbf{T}\hat{\mathbf{\Lambda}}_{G}\mathbf{T}\mathbf{1}_{f,k}\right] + \sigma_{\xi,f,k}^{2}\sigma_{\xi,g,k^{'}}^{2}\mbox{Tr}\left\{\mathbb{E}\left[\hat{\mathbf{\Lambda}}_{F}^{*}\mathbf{T}\hat{\mathbf{\Lambda}}_{G}\mathbf{T}\right]\right\}\nonumber\\
&\stackrel{(c)}{=}\frac{\sigma_{\xi,f,k}^{2}}{(N-2K-1)\hat\sigma_{f,k}^{2}}+ \frac{\sigma_{\xi,g,k^{'}}^{2}}{(N-2K-1)\hat\sigma_{g,k^{'}}^{2}}+ \sigma_{\xi,f,k}^{2}\sigma_{\xi,g,k^{'}}^{2}\hat{\eta}\label{varzf}
\setcounter{equation}{35}
\end{align}
\setcounter{equation}{28}
\hrule
\vspace{-0.25in}
\end{figure*}
}

Starting with the numerator of (\ref{gammalower}),  we have
\begin{align}\setcounter{equation}{27}
\mathbb{E}\left[{\mathbf{f}_{k}}^{T}\mathbf{W}\mathbf{g}_{k^{'}}\right] &= \mathbb{E}\left[\left(\hat{\mathbf{f}}_{k}+\mathbf{e}_{f,k}\right)^{T}\hat{\bar{\mathbf{F}}}^{*}\mathbf{T}\hat{\bar{\mathbf{G}}}^{H}\left(\hat{\mathbf{g}}_{k^{'}}+\mathbf{e}_{g,k^{'}}\right)\right]\nonumber\\
&\stackrel{(a)}{=} \mathbb{E}\left[{\mathbf{1}}_{f,k}^{T}\mathbf{T}{\mathbf{1}}_{g,k^{'}}\right]\stackrel{(b)}{=} \mathbb{E}[1] = 1,
\end{align}
The equality in $(a)$ is obtained by using the following results: $\hat{\mathbf{g}}_{k^{'}}^{H}\hat{\bar{\mathbf{G}}}=\mathbf{1}_{g,k^{'}}^{T}$, $\hat{\bar{\mathbf{G}}}^{H}\hat{\mathbf{g}}_{k^{'}}=\mathbf{1}_{g,k^{'}}$, $\hat{\mathbf{f}}_{k}^{T}\hat{\bar{\mathbf{F}}}^{*}=\mathbf{1}_{f,k}^{T}$, $\hat{\bar{\mathbf{F}}}^{T}\hat{\mathbf{f}}_{k}^{*}=\mathbf{1}_{f,k}$. Equality in $(b)$ is because ${\mathbf{1}}_{f,k}^{T}\mathbf{T}{\mathbf{1}}_{g,k^{'}}=1$. The expression for $\mbox{var}\left[{\mathbf{f}_{k}}^{T}\mathbf{W}\mathbf{g}_{k^{'}}\right]$ is given by \iftoggle{SINGLE_COL}{}{(\ref{varzf}) (solved at the top of this page).}
{The equality in $(a)$ therein is because $\hat{\mathbf{F}}^{T}\mathbf{W}\hat{\mathbf{G}}= \mathbf{T},\,\mbox{i.e.},\, \hat{\mathbf{f}}_{k}^{T}\mathbf{W}\hat{\mathbf{g}}_{j}= 1,\forall j = k^{'}\, \mbox{and}\, 0, \, \mbox{otherwise}$. Equalities in $(b)$ are obtained by substituting the value of $\mathbf{W}$ from \eqref{wzf} and by simple manipulations. The equality in $(c)$ is because $\mathbb{E}\left[\hat{w}_{f,k,k}\right] =\frac{1}{(N-2K-1)\hat\sigma_{f,k}^{2}}$, $\mathbb{E}\left[\hat{w}_{g,k^{'},k^{'}}\right] =\frac{1}{(N-2K-1)\hat\sigma_{g,k^{'}}^{2}}$, and $\hat{\eta}\triangleq\sum_{j=1}^{2K}\frac{1}{\left(N-2K-1\right)^{2}\hat\sigma_{f,j}^{2}\hat\sigma_{g,	j^{'}}^{2}}
$.
Remember that there is no need to perform SIC in the case ZFR/ZFT processing, and hence the self-interference term can be re-written as
\begin{align} \setcounter{equation}{29}
&\mbox{SI}_{k} = \frac{1}{(N-2K-1)}\hspace{-0.03in}\left(\hspace{-0.03in}\frac{\sigma_{\xi,f,k}^{2}}{\hat\sigma_{f,k^{'}}^{2}}\hspace{-0.03in}+\hspace{-0.03in} \frac{\sigma_{\xi,g,k}^{2}}{\hat\sigma_{g,k^{'}}^{2}}\hspace{-0.03in}\right)\hspace{-0.03in}+\hspace{-0.03in} \sigma_{\xi,f,k}^{2}\sigma_{\xi,g,k}^{2}\hat{\eta}.\label{sizf}
\end{align}
Similarly, the other terms in the denominator of (\ref{gammalower}) can be written as follows
\begin{align}
&\hspace{-0.08in}\mbox{IP}_{k}\hspace{-0.03in}=\hspace{-0.1in} \sum_{i\ne k,k^{'}}^{2K}\hspace{-0.08in}p_{i}\hspace{-0.03in}\left[\hspace{-0.03in}\frac{1}{N\hspace{-0.03in}-\hspace{-0.03in}2K\hspace{-0.03in}-\hspace{-0.03in}1}\hspace{-0.03in}\left(\hspace{-0.03in}\frac{\sigma_{\xi,f,k}^{2}}{\hat\sigma_{f,i^{'}}^{2}}\hspace{-0.03in}+\hspace{-0.03in} \frac{\sigma_{\xi,g,i}^{2}}{\hat\sigma_{g,k^{'}}^{2}}\hspace{-0.03in}\right)\hspace{-0.03in}+\hspace{-0.03in} \sigma_{\xi,f,k}^{2}\sigma_{\xi,g,i}^{2}\hat{\eta}\right]\hspace{-0.05in},\label{ipzf}\\
&\mbox{NR}_{k}=  \sigma_{nr}^{2}\left(\frac{1}{(N-2K-1)\hat\sigma_{g,k^{'}}^{2}}+ \sigma_{\xi,f,k}^{2}\hat{\eta}\right)\iftoggle{SINGLE_COL}{}{,\label{NR}\\
&}\iftoggle{BIG_EQUATION}{}{,\quad}\mbox{LIR}_{k}= P_{R}\sigma_{LIR}^{2}\left(\frac{1}{(N-2K-1)\hat\sigma_{g,k^{'}}^{2}}+ \sigma_{\xi,f,k}^{2}\hat{\eta}\right),\label{lirzf}\\
&\mbox{UI}_{k}= \sum_{i,k\in U{k}}p_{i}\sigma_{k,i}^{2},\,
\mbox{NU}_{k}= \sigma_{n}^{2}.\label{nuzf}
\end{align}
Substituting the values obtained from (\ref{varzf}-\ref{nuzf}) in  (\ref{gammalower}), we obtain~\eqref{gammazft2}.
\vspace{-0.1in}
\bibliographystyle{IEEEtran}
\bibliography{IEEEabrv,Relay_ref}

\end{document}